\documentclass[submission,copyright,creativecommons]{eptcs}
\providecommand{\event}{GandALF 2018} 

\usepackage[utf8]{inputenc}
\usepackage[australian]{babel}
\usepackage[final]{microtype}
\usepackage{booktabs}
\usepackage{pifont}  

\newcommand\fitpar{\looseness=-1}

\usepackage{underscore}

\usepackage{cite}
\usepackage{amsmath,amssymb,amsthm}
\usepackage{mathtools,thmtools}
\usepackage[capitalise]{cleveref}
\usepackage[all]{foreign}
\usepackage[capitalise]{cleveref}
\usepackage{enumitem}

\usepackage{tikz}

\newlist{claims}{enumerate*}{1}
\setlist[claims]{label={(\roman*)}}

\usepackage{temporallogic}

\RenewDocumentCommand\mathbinaryop{}{\mathsf}

\newlist{rules}{description}{2}
\setlist[rules]{format={\Rule}}

\DeclareMathAlphabet{\mathcal}{OMS}{cmsy}{m}{n}

\declaretheorem{theorem}
\declaretheorem[sibling=theorem]{proposition}
\declaretheorem[sibling=theorem]{lemma}
\declaretheorem[sibling=theorem]{definition}

\newcommand\DeclareSymbol[2][\mathcal]{\csdef{#2}{\ensuremath{#1{#2}}\xspace}}

\DeclareSymbol[\mathsf]{V}
\DeclareSymbol[\mathsf]{AP}
\DeclareSymbol[\mathsf]{K}
\DeclareSymbol[\mathbb]{N}
\DeclareSymbol[\mathbb]{Z}
\DeclareSymbol[\mathcal]{C}
\DeclareSymbol[\mathcal]{E}
\DeclareSymbol[\mathsf]{wwf}

\DeclarePairedDelimiter\seq{\langle}{\rangle}
\DeclarePairedDelimiter\pair{\langle}{\rangle}
\DeclarePairedDelimiter\set{\{}{\}}
\let\time\undefined
\DeclareMathOperator\time{\mathsf{time}}
\DeclareMathOperator\delay{\mathsf{d}}

\renewcommand\bar[1]{\overline{#1}}

\DeclareDocumentCommand\ticked{}{\text{\ding{51}}}
\DeclareDocumentCommand\crossed{}{\text{\ding{55}}}

\allowdisplaybreaks

\usepackage{lipsum}
\usepackage{xcolor}

\newcommand\suitcase{
  \begin{tikzpicture}
    \fill[rounded corners=1pt] (0,0) rectangle (1em,0.75em);
    \draw[white] (0.2em,0) -- (0.2em,0.75em);
    \draw[white] (0.8em,0) -- (0.8em,0.75em);
    \draw[line width=0.5pt, rounded corners=0.1pt]
      (0.34em, 0.75em) -- (0.34em, 0.85em) -| (0.66em, 0.75em);
  \end{tikzpicture}
}

\newcommand\blthanks[2][]{%
  \begingroup
  \renewcommand\thefootnote{}\footnotetext{%
    \makebox[0pt][r]{%
      \scriptsize \raisebox{-0.5pt}{#1}\hspace{0.2em}%
    }#2%
  }%
  \addtocounter{footnote}{-1}%
  \endgroup
}

\title{%
  One-Pass and Tree-Shaped Tableau Systems\\
   for TPTL and TPTL$_{\mbox{\large b}}$+Past%
}

\author{%
  Luca Geatti, Nicola Gigante, and Angelo Montanari
  \institute{University of Udine, Italy}
  \email{\{geatti.luca,gigante.nicola\}@spes.uniud.it}
  \email{angelo.montanari@uniud.it}
  \and Mark Reynolds
  \institute{The University of Western Australia}
  \email{mark.reynolds@uwa.edu.au}
}
\def\titlerunning{One-pass and Tree-Shaped Tableau
                  Systems for TPTL and TPTL$_{\mbox{\scriptsize b}}$+Past}
\def\authorrunning{L.~Geatti, N.~Gigante, A.~Montanari, and M.~Reynolds}

\begin{document}
\maketitle

\begin{abstract}
  In this paper, we propose a novel one-pass and tree-shaped tableau method for Timed
  Propositional Temporal Logic and for a bounded variant of its extension with past
  operators. Timed Propositional Temporal Logic (\TPTL) is a \emph{real-time} temporal logic,
  with an \EXPSPACE-complete satisfiability problem, which has been successfully
  applied to the verification of real-time systems. In contrast to \LTL, adding
  past operators to \TPTL makes the satisfiability problem for the resulting
  logic  (\TPTLP) \emph{non-elementary}. In this paper, we devise a one-pass and
  tree-shaped tableau for both \TPTL and \emph{bounded} \TPTLP (\TPTLbP), a syntactic
  restriction introduced to encode timeline-based planning problems, which
  recovers the \EXPSPACE-complete complexity. The tableau systems for \TPTL and
  \TPTLbP are presented in a unified way, being very similar to each other,
  providing a common skeleton that is then specialised to each logic. In doing
  that, we characterise the semantics of \TPTLbP in terms of a purely syntactic
  fragment of \TPTLP, giving a translation that embeds the former into the
  latter. Soundness and completeness of the system are proved fully. In particular,
  we give a greatly simplified model-theoretic completeness proof, which
  sidesteps the complex combinatorial argument used by known proofs for the
  one-pass and tree-shaped tableau systems for \LTL and \LTLP.
\blthanks[\suitcase]{
  These results were developed mainly while A.\ Montanari was on leave at the
  \emph{Stockholm University}, and N.\ Gigante was on leave at the
  \emph{University of Western Australia}, supported by the \emph{AIxIA Outgoing
  mobility grant 2017}. The work was partially supported by the Italian GNCS
  project \emph{Formal methods for verification and synthesis of discrete and
  hybrid systems} (N.~Gigante and A.~Montanari), the PRID project \emph{ENCASE -
  Efforts in the uNderstanding of Complex interActing SystEms} (N.~Gigante and
  A.~Montanari), and the Australian Research Council funding--DP140103365 (M.
  Reynolds).
}
\end{abstract}


\section{Introduction}
\label{sec:introduction}

Among the reasoning methods used to decide the satisfiability of logical formulae,
\emph{tableau methods} are among the earliest proposed and most studied solutions~\cite{Tableau99}. Classic tableau methods for logics of the linear time, such as, for instance, Linear
Temporal Logic (\LTL)~\cite{LichtensteinP00,MannaPnueli95}, build a graph
structure which is then traversed to look for possible models of the formula. Despite
being a useful theoretical tool, such graph-shaped tableau systems are not efficient in
practice as they need to build and traverse a huge graph structure in multiple passes.
Various ways of overcoming such a limitation have been proposed in the literature, including
incremental~\cite{KestenMMP93} and single pass techniques~\cite{Schwendimann98}.
Recently, a one-pass and tree-shaped tableau system for \LTL has been
devised~\cite{Reynolds16a}, which does not build any huge preliminary
structure and, thanks to its pure rule-based tree-search structure, proved to be
amenable to efficient implementation and easy
parallelisation~\cite{McCabeDanstedR17,BertelloGMR16}. Recent work also suggested
that its modular structure makes it possible to easily extend it to other linear time
logics (the extension to \LTL with \emph{past operators} is described in~\cite{GiganteMR17}).

Timed Propositional Temporal Logic (\TPTL) is a linear time logic, which extends
\LTL with the ability to express real-time properties of systems and
computations~\cite{AlurH94}. The greater expressive power of \TPTL is reflected
in the computational complexity of its satisfiability problem, which is
\EXPSPACE-complete. Originally proposed as a formal tool for the verification of
real-time systems, it recently found interesting applications in the area of
artificial intelligence, to encode a meaningful class of \emph{timeline-based
planning} problems~\cite{DellaMonicaGMSS17}.
This and other application scenarios benefit from/require the use of \emph{past
operators}, which allow the logic to compactly predicate about events 
in the past of the current time point. However, in contrast to the case of \LTL,
where past operators can be supported without harm, adding them to \TPTL greatly
increases the complexity of its satisfiability problem, which becomes
\emph{non-elementary}~\cite{AlurH93}.
For this reason, \emph{bounded TPTL with Past} (\TPTLbP) has been
introduced~\cite{DellaMonicaGMSS17}, which supports past operators, but suitably
restricts their use in order to recover an \EXPSPACE-complete satisfiability
problem. While initially introduced as a specific tool to encode planning
problems, \TPTLbP is interesting by itself, since it enables the use of past
operators in a fairly natural way.

In this paper, we exploit the extensibility of the aforementioned tableau system
to provide a one-pass and tree-shaped tableau method for \TPTL and \TPTLbP.
We present both tableau systems, which are very similar, in a unified way by first
\begin{claims}
  \item factoring out the common structure, and then
  \item showing how to specialise it in the case of
        \TPTL (future-only) and \TPTLbP (bounded) formulae, thus obtaining a
        one-pass and tree-shaped tableau system for both logics.
        To show how the tableau for \TPTLbP formulae works,
  \item we characterise the semantics of the logic in terms of a guarded
        fragment of the full \TPTLP logic, showing how to translate \TPTLbP into
        this fragment. Furthermore,
  \item the completeness of the two tableau systems is shown by a greatly simplified
        proof exploiting a new model-theoretic technique which sidesteps the
        complex combinatorial argument used by known proofs for \LTL and
        \LTLP. 
\end{claims}

The tableau systems presented here for \TPTL and \TPTLbP are truly extensions of
the previously known ones for \LTL and \LTLP, respectively, in the sense that
their rules and behaviour are exactly the same as before when applied to pure
\LTL/\LTLP formulae, further confirming the modular and extensible nature of the
one-pass tree-shaped system.

The paper is structured as follows. Syntax and semantics of \TPTL, \TPTLP, and
\TPTLbP are illustrated in \cref{sec:tptl}. Then, \cref{sec:tableau} describes
the tableau systems for \TPTL and \TPTLbP. It first introduces the general skeleton
common to both, and then it shows how to tailor it to \TPTL and \TPTLbP. Finally,
soundness and completeness of both systems are proved in \cref{sec:proofs}.
\cref{sec:conclusions} concludes with  some considerations on the obtained
results and open problems.


\section{Timed Propositional Temporal Logic}
\label{sec:tptl}

This section defines syntax and semantics of \TPTL~\cite{AlurH94}, \TPTLP~\cite{AlurH94}, and \TPTLbP~\cite{DellaMonicaGMSS17}.
Let $\AP=\{p,q,r,\ldots\}$ be a set of \emph{proposition letters} and
$\V=\{x,y,z,\ldots\}$ be a set of \emph{variables}. A \TPTLP formula $\phi$
over $\AP$ and $\V$ is recursively defined as follows:
\begin{align*}
  \phi := p & {} \mid \neg\phi_1 \mid \phi_1\lor\phi_2 \mid x.\phi_1
                 \mid x \le y + c \mid x \le c \mid x \equiv_m y + c \\
 & {} \mid \X \phi_1 \mid \phi_1\U\phi_2 \mid \phi_1\R\phi_2
                 \mid \Y \phi_2 \mid \phi_1\S\phi_2 \mid \phi_1\T\phi_2,
\end{align*}
where $p\in\AP$, $\phi_1$ and $\phi_2$ are \TPTLP formulae, $x,y\in\V$, $c\in\Z$, $m\in\N$, and $\equiv_m$ is the congruence modulo the constant $m$. Formulae of the form $x.\phi$ are called \emph{freeze
quantifications}, while those of the forms $x\le y + c$, $x\le c$, and $x\equiv_m
y$ are called \emph{timing constraints}. Standard logical and temporal shortcuts,
e.g., $\true$ for $p\lor\neg p$, for some $p\in\AP$, $\false$ for $\neg\true$,
$\phi_1\land\phi_2$ for $\neg(\neg\phi_1\lor\neg\phi_2)$, $\F\phi$ for $\true\U\phi$,
$\G\phi$ for $\neg\F\neg\phi$, and $\P\phi$ for $\true\S\phi$, as well as constraint
shortcuts, e.g., $x \le y$ for $x \le y + 0$, $x > y$ for $\neg(x \le y)$, and $x=y$ for
$\neg(x<y)\land\neg(y<x)$, are used. A formula $\phi$ is \emph{closed} if each occurrence of
a variable $x$ is enclosed by a subformula of the form $x.\psi$. As for \LTLP,
the temporal operators can be partitioned in \emph{future}
(\emph{tomorrow} $\X$, \emph{until} $\U$, and \emph{release} $\R$) and
\emph{past} (\emph{yesterday} $\Y$, \emph{since} $\S$, and
\emph{triggered} $\T$) ones. \TPTL is the fragment of \TPTLP where only future operators are used.

\TPTLP formulae are interpreted over \emph{timed state sequences}, i.e.,
structures $\rho=(\sigma,\tau)$, where $\sigma=\seq{\sigma_0,\sigma_1,\ldots}$
is an infinite sequence of states $\sigma_i\in2^\AP$, for $i \ge 0$, and
$\tau=\seq{\tau_0,\tau_1,\ldots}$ is an infinite sequence of \emph{timestamps}
$\tau_i\in\N$, for $i \ge 0$, such that (i) $\tau_{i+1}\ge\tau_i$
(\emph{monotonicity}), and (ii) for all $t\in\N$, there is some $i\ge0$ such that
$\tau_i\ge t$ (\emph{progress}).
Formally, the semantics of \TPTLP is defined as follows. Functions $\xi:\V\to\N$
mapping variables to timestamps are called \emph{environments}. A timed state
sequence $\rho$ \emph{satisfies} a formula $\phi$ at position $i\ge0$, with
environment $\xi$, written $\rho\models^i_\xi\phi$, if (and only if):

\begin{conditions}
\item $\rho \models^i_\xi p$                & $p \in \sigma_i$;
\item $\rho \models^i_\xi \phi_1\lor\phi_2$ &
       $\rho \models^i_\xi\phi_1$ or $\rho \models^i_\xi \phi_2$; 
\item $\rho \models^i_\xi \neg\phi_1$       & $\rho \not\models^i_\xi\phi_1$;
\item $\rho \models^i_\xi x \le y + c$      &  $\xi(x) \le \xi(y) + c$;
\item $\rho \models^i_\xi x \le c$          &  $\xi(x) \le c$;
\item $\rho \models^i_\xi x \equiv_m y + c$ &  $\xi(x) \equiv_m \xi(y) + c$;
\item $\rho \models^i_\xi x.\phi_1$         &
      $\rho \models^i_{\xi'} \phi_1$ where $\xi'=\xi[x\leftarrow \tau_i]$;
\item $\rho \models^i_\xi \X\phi_1$       & $\rho\models^{i+1}_\xi \phi_1$;
\item $\rho \models^i_\xi \phi_1\U\phi_2$ &
      there exists $j\ge i$ such that $\rho \models^j_\xi \phi_2$ and
      $\rho \models^k_\xi \phi_1$ for all $i\le k < j$;
\item $\rho \models^i_\xi \phi_1\R\phi_2$ & either $\rho\models^j_\xi\phi_2$
      for all $j\ge i$, or there exists a $k\ge i$ such that
      $\rho\models^k_\xi\phi_1$ and $\rho\models^j_\xi\phi_2$ for all
      $i\le j\le k$;
\item $\rho \models^i_\xi \Y\phi_1$       &
      $i > 0$ and $\rho^{i-1}\models_\xi\phi_1$;
\item $\rho \models^i_\xi \phi_1\S\phi_2$ &
      there exists $j\le i$ such that $\rho \models^j_\xi \phi_2$, and
      $\rho \models^k_\xi \phi_1$ for all $j < k \le i$;
\item $\rho \models^i_\xi \phi_1\T\phi_2$ & either $\rho\models^j_\xi\phi_2$
      for all $0\le j\le i$, or there exists a $k\le i$ such that
      $\rho\models^k_\xi\phi_1$ and $\rho\models^j_\xi\phi_2$ for all
      $i\ge j\ge k$,
\end{conditions}
\noindent where $\xi'=\xi[x\leftarrow\tau_i]$ is the environment equal to $\xi$
possibly excepting
$\xi'(x)=\tau_i$. A \emph{closed} formula $\phi$ is satisfied by
a timed state sequence $\rho$, written $\rho\models\phi$, if $\rho \models^0_\xi \phi$,
for any $\xi$. \TPTL and \TPTLP can thus be viewed as (\emph{metric}) extensions of,
respectively, \LTL and \LTLP with
the \emph{freeze quantifier}
$x.\phi$, that allows one to bind a variable to  the \emph{timestamp} of the
current state, which can then be compared with other variables by the timing
constraints.
In contrast to \LTL and \LTLP, which
both have a \PSPACE-complete satisfiability problem, adding past operators to
\TPTL causes a complexity blowup: the satisfiability problem is
\EXPSPACE-complete for \TPTL, but \emph{non-elementary} for
\TPTLP~\cite{AlurH94}.

The unconstrained use of past operators in these timed logics is thus impossible
in practice. However, there are many scenarios where referring to the past may
be needed, and thus it is useful to search for possible ways of adding past
operators to \TPTL while retaining a (relatively) practicable complexity.
\TPTLbP has been introduced
to encode a meaningful class of timeline-based planning problems, whose synchronisation
rules can interchangeably refer to the future or the past~\cite{DellaMonicaGMSS17}.
The syntax of \TPTLbP is similar to that of \TPTLP, the only difference being that
each temporal operator is subscripted with a bound which constrains the visibility of
the operator. Formally, a \TPTLbP formula $\phi$ over $\AP$ and $\V$ is
recursively defined as follows:
\begin{align*}
  \phi := p & {} \mid \neg\phi_1 \mid \phi_1\lor\phi_2 \mid x.\phi_1
                 \mid x \le y + c \mid x \le c \mid x \equiv_m y + c\\
            & {} \mid \X_w \phi_1 \mid \wX_w \phi_1 \mid \phi_1\U_w\phi_2
                 \mid \phi_1\R_w\phi_2 
            \mid \Y_w \phi_2 \mid \wY_w \phi_2 \mid \phi_1\S_w\phi_2 \mid
                 \phi_1\T_w\phi_2,
\end{align*}
where $w\in\N\cup\{+\infty\}$, $p\in\AP$, $\phi_1, \phi_2$
are \TPTLbP formulae, $x,y\in\V$, $m\in\N$, and $c\in\Z$. The bound on any
temporal operator can be $w=+\infty$ (or omitted) only if applied to a
\emph{closed} formula.
This restriction limits any temporal modality (including future ones) to look
only as far as their bound. As it will be shown later, this implies that when
interpreting any timing constraint, such as, e.g., $x\le y + c$, the timestamps
$x$ and $y$ can be distant, at most, an amount of time which is exponential in
the size of the formula.

Formally, the semantics of \TPTLbP is defined as follows. Let $\rho$ be a timed state
sequence and let $\xi$ be an environment. We say that $\rho$ satisfies a \TPTLbP formula
$\phi$ at position $i\ge0$ with environment $\xi$, written $\rho\models^i_\xi\phi$,
if (and only if):

\begin{conditions}
\item $\rho \models^i_\xi \X_w\phi_1$       &
      $\tau_{i+1}-\tau_i\le w$ and $\rho\models^{i+1}_\xi \phi_1$;
\item $\rho \models^i_\xi \wX_w\phi_1$       &
      $\tau_{i+1}-\tau_i\le w$ implies $\rho\models^{i+1}_\xi \phi_1$;
\item $\rho \models^i_\xi \phi_1\U_w\phi_2$ &
      there exists $j\ge i$ such that: (i) $\tau_j-\tau_i\le w$,
      (ii) $\rho \models^j_\xi \phi_2$,\newline
      and (iii) $\rho \models^k_\xi \phi_1$ for all $i\le k < j$;
\item $\rho \models^i_\xi \phi_1\R_w\phi_2$ & either (i) $\tau_j-\tau_i\le w$
      implies $\rho\models^j_\xi\phi_2$ for all $j\ge i$, or
      (ii) there exists\newline
      $k\ge i$ such that $\tau_k-\tau_i\le w$ and $\rho\models^k_\xi\phi_1$,
      and $\rho\models^j_\xi\phi_2$ for all $i\le j\le k$;
\item $\rho \models^i_\xi \Y_w\phi_1$       &
      $i > 0$, $\tau_i-\tau_{i-1}\le w$, and $\rho \models^{i-1}_\xi\phi_1$;
\item $\rho \models^i_\xi \wY_w\phi_1$       &
      $i > 0$ and $\tau_i-\tau_{i-1}\le w$ imply $\rho \models^{i-1}_\xi\phi_1$;
\item $\rho \models^i_\xi \phi_1\S_w\phi_2$ &
      there exists $j\le i$ such that: (i) $\tau_i-\tau_j\le w$,
      (ii) $\rho \models^j_\xi \phi_2$,\newline
      and (iii) $\rho \models^k_\xi \phi_1$ for all $j < k \le i$;
\item $\rho \models^i_\xi \phi_1\T_w\phi_2$ & either (i) $\tau_i-\tau_j\le w$
      implies $\rho\models^j_\xi\phi_2$ for all $0\le j\le i$, or (ii) there
      exists $k\le i$ such that $\tau_i-\tau_k\le w$ and $\rho\models^k_\xi\phi_1$,
      and $\rho\models^j_\xi\phi_2$ for all $k\le j\le i$;
\item \nocondition{same semantics as \TPTLP for the remaining operators.}
\end{conditions}

In addition to the bounded versions of all the temporal
operators of \TPTLP, \TPTLbP includes a \emph{weak} version of both the
\emph{tomorrow} and \emph{yesterday} ones. While the formula $\X_w\phi$
(resp., $\Y_w\phi$) require the next (resp., previous) state to be distant at
most $w$ time steps and to satisfy $\phi$, the \emph{weak tomorrow}
(resp., \emph{yesterday}) operator in a formula of the form $\wX\phi$
(resp., $\wY\phi$), requires the next (resp., previous) state to satisfy
$\phi$ \emph{only if} such a state exists and its distance is at most $w$. The weak
tomorrow and yesterday operators are introduced as \emph{duals} of the standard ones,
in such a way that $\neg\X_w\phi\equiv\wX_w\neg\phi$ and $\neg\wX_w\phi\equiv\X_w\neg\phi$
(and similarly for the yesterday ones).
This ensures that each temporal modality has its own negated dual (such as the
\emph{until}/\emph{release} and \emph{since}/\emph{triggered} pairs), so that
any \TPTLbP formula can be put into \emph{negated normal form}, where negations
are only applied to proposition letters and timing constraints. The existence of a
negated normal form for \TPTLbP formulae will play an important role in the definition
of the tableau system (see \cref{sub:tptlbp-tableau}).


\section{The tableau systems for \TPTL and \TPTLbP}
\label{sec:tableau}

This section describes the one-pass and tree-shaped tableau systems for \TPTL and
\TPTLbP, that respectively extend those for \LTL and \LTLP presented in
\cite{GiganteMR17,DellaMonicaGMSS17}. Soundness and completeness of the systems
are proved in \cref{sec:proofs}. The two systems are very similar, differing only
in specific parts and sharing the vast majority of their workings. Hence, a
common skeleton is first described, making some assumptions that will then be
fulfilled for the two specific logics.

\subsection{The common skeleton}
\label{subsec:skeleton}

The parts in common between the two tableau systems will be presented as if they
were supposed to handle \TPTLP formulae. \TPTL is a proper fragment of \TPTLP,
and \TPTLbP, as it will be shown later, can be fully embedded in a proper guarded
fragment of \TPTLP. Hence, both tableaux do indeed handle \TPTLP formulae, albeit
of a specific kind. We will mention the specific logics when stating results
that are not proved for the full \TPTLP logic.

W.l.o.g, we may assume formulae to be in \emph{negated normal form}, which is
guaranteed to exist for formulae of both logics. As shown in \cite{AlurH94} for
\TPTL and in \cite{DellaMonicaGMSS17} for \TPTLbP, w.l.o.g., we can also restrict
ourselves to models with a bound on the maximum \emph{temporal} distance between
two subsequent states.

\begin{proposition}[$\delta$-bounded models~\cite{AlurH94,DellaMonicaGMSS17}]
  \label{prop:delta-bounded-models}
  Let $\phi$ be a closed \TPTL or \TPTLbP formula. A model
  $\rho=\pair{\sigma,\tau}$ of $\phi$ is said to be \emph{$\delta$-bounded}, for
  some $\delta\ge0$, if $\tau_{i+1}-\tau_i\le\delta$ for all $i\ge0$. Then,
  it holds that $\phi$ is satisfiable if and only if there exists some
  $\delta_{\phi}\ge0$ such that $\phi$ has a $\delta_{\phi}$-bounded model.
\end{proposition}

In \cite{AlurH94,DellaMonicaGMSS17}, it is shown how to compute $\delta_\phi$
starting from the constants appearing in
$\phi$: roughly, $\delta_\phi$
is the product of all the constants in $\phi$.
Similarly, we can assume that no \emph{absolute} timing constraints (those of
the form $x\le c$) are used in the formulae (see Lemma 6 in \cite{AlurH94}).
W.l.o.g., we can also assume that any variable $x$ is used only in one
freeze quantifier in any formula, so that in a formula like $x.\psi$ any occurrence
of $x$ in $\psi$ is free. Since freeze quantifiers can be pushed out of boolean
connectives, when talking about closed formulae we will write them as $x.\psi$,
with explicit reference to the outermost freeze quantifier.

We start by defining an important
building block of the system.

\begin{definition}[Temporal shift]
  \label{def:temporal-shift}
  Let us denote as \wwf the set of all the well-formed \TPTLP formulae. The
  \emph{temporal shift operator} is a function
  $\cdot^\delta:\wwf\times\Z\to\wwf$ such that:
  \begin{enumerate}
  \item \label{def:temporal-shift:semantics}
        for any closed \TPTLP formula $x.\psi$ and any $\delta\in\Z$, timed
        state sequence $\rho$, environment $\xi$, and position $i\ge0$, it holds
        that $\rho\models^i_\xi x.\psi^\delta$ if and only if
        $\rho\models^i_{\xi'}\psi$, where
        $\xi'=\xi[x\leftarrow \tau_i-\delta]$;
  \item \label{def:temporal-shift:convergence}
        there exists $\delta'\in\Z$ such that
        $x.\psi^{\delta'}= x.\psi^{\delta'+i}$ and
        $x.\psi^{-\delta'}= x.\psi^{-\delta'-i}$ for all $i\ge0$.
  \end{enumerate}
\end{definition}

\Cref{def:temporal-shift:semantics} of \cref{def:temporal-shift} states that the truth value of $x.\psi^\delta$, interpreted at the current
state, is the same as that of $\psi$ in the case where $x$ were
bound to the timestamp of a previous state located exactly $\delta$ time units before.
By \cref{def:temporal-shift:convergence}, this transformation has to be defined in such a way that it converges to a fixed point after a large enough amount of shifting, so that for a given $x.\psi$, the number of different formulae of the form $x.\psi^\delta$ is finite. It is not
known whether such an operator exists for full \TPTLP.
Later, we will show how to define it in the cases of \TPTL and
\TPTLbP.

The \emph{closure} of a formula $z.\phi$ contains all the formulae that are
relevant to the satisfaction of $z.\phi$.
\begin{definition}[Closure of a formula]
  \label{def:closure}
  Let $z.\phi$ be a closed \TPTLP formula and let $\cdot^\delta$ be a temporal
  shift operator. Then, the \emph{closure} of $z.\phi$ is the set
  $\C(z.\phi)$ recursively defined as follows:
  \begin{enumerate}
    \item \label{def:closure:base}
          $z.\phi \in\C(z.\phi)$;
    \item \label{def:closure:and}
          if $x.(\psi_1\land\psi_2)\in\C(z.\phi)$, then
          $\set{x.\psi_1,x.\psi_2}\subseteq\C(z.\phi)$;
    \item \label{def:closure:or}
          if $x.(\psi_1\lor\psi_2)\in\C(z.\phi)$, then
          $\set{x.\psi_1,x.\psi_2}\subseteq\C(z.\phi)$;
    \item \label{def:closure:tomorrow}
          if $x.\X\psi\in\C(z.\phi)$, then
          $x.\psi^\delta\in\C(z.\phi)$, for all $\delta\ge0$;
    \item \label{def:closure:yesterday}
          if $x.\Y\psi\in\C(z.\phi)$, then
          $x.\psi^{-\delta}\in\C(z.\phi)$, for all $\delta\ge0$;
    \item \label{def:closure:operators}
          if $x.(\psi_1\circ\psi_2)\in\C(z.\phi)$, where
          $\circ\in\set{\U,\R,\S,\T}$, then
          $\set{x.\psi_1,x.\psi_2,x.\X(\psi_1\circ\psi_2)}\subseteq\C(z.\phi)$;
    \item if $x.y.\psi\in\C(z.\phi)$, then $x.\psi[y/x]\in\C(z.\phi)$.
  \end{enumerate}
\end{definition}

Note that, if $\cdot^\delta$ is a temporal shift operator, then $\C(z.\phi)$ is
a finite set, thanks to \cref{def:temporal-shift:convergence} of
\cref{def:temporal-shift}. Moreover, note that, by construction, every formula
in $\C(z.\phi)$ is a \emph{closed} formula.

Now we can effectively start describing the one-pass and tree-shaped tableau system
for \TPTLP. The tableau for a closed formula $z.\phi$ is a tree where each node
$u$ of the tree is labelled with a \emph{finite} set 
$\Gamma(u)\subseteq\C(z.\phi)$. 
Additionally, a non-negative
integer $\time(u)\in\N$ is associated with each node $u$. Given two nodes $u$ and
$v$, we write $u\le v$ ($u<v$) if $u$ is a (proper) ancestor of $v$. The root
note $u_0$ is labelled by the formula itself, i.e., $\Gamma(u_0)=\set{z.\phi}$, and is
set at $\time(u_0)=0$. The tableau is built top-down, from the root to the
leaves, performing a state-by-state search for a model of the formula where each
accepted branch of the complete tableau corresponds to a satisfying model. At
each step, a set of \emph{expansion rules} is applied to the leaf nodes of the
tree, until no expansion rule can be applied anymore. Each application of an expansion rule
results in the addition of one or more children to the selected node, making the
tree grow and refining the choice of which formulae of the closure have to hold
at the current state. Then, a set of \emph{termination rules} decides if the
current tableau branch has to be \emph{accepted}~(\ticked),
\emph{rejected}~(\crossed), or if the branch can continue to be explored, making
a step to the next state. Expansion rules are shown in
\cref{tab:expansion-rules}. Each rule of the form $\psi\to\Delta'$ is applied to
any node $u$ such that $\psi\in\Gamma(u)$ and causes the addition of a child
$u'$ of $u$ such that $\Gamma(u')=(\Gamma(u)\setminus\set{\psi})\cup\Delta'$.
Similarly, a rule of the form $\psi\to\Delta'\mid\Delta''$ causes the addition
of two children $u'$ and $u''$, where $\Gamma(u')=(\Gamma(u)\setminus\set{\psi})
\cup\Delta'$ and $\Gamma(u'')=(\Gamma(u)\setminus\set{\psi})\cup\Delta''$.

By construction, repeatedly applying expansion rules will eventually result into
leaves labelled only by proposition letters, timing constraints, or formulae of the
forms $x.\X\psi$ or $x.\Y\psi$, which cannot be further expanded. Formulae of
this kind are called \emph{elementary formulae}, and a node (resp., a leaf)
whose label contains only elementary formulae is a \emph{poised node} (resp.,
poised leaf).

\begin{table}
  \centering
  \begin{tabular}{>{\rulefont}r  l !{$\to$} l}
		\toprule
			\textterm{Name} & \multicolumn{2}{l}{\textterm{Rule}} \\ \midrule
			CONJUNCTION  & $x.(\psi_1 \land \psi_2)$ & $\set{x.\psi_1, x.\psi_2}$ \\
      FREEZE       & $x.y.\psi_1$
                   & $\set{x.\psi_1[y / x]}$ \\
			DISJUNCTION  & $x.(\psi_1 \lor \psi_2)$
                   & $\set{x.\psi_1} \mid \set{x.\psi_2}$ \\
			UNTIL        & $x.(\psi_1 \U \psi_2)$
                   & $\set{x.\psi_2}\mid\set{x.\psi_1, x.\X(\psi_1 \U \psi_2)}$ \\
      SINCE        & $x.(\psi_1 \S \psi_2)$
                   & $\set{x.\psi_2} \mid \set{x.\psi_1, x.\Y(\psi_1\S\psi_2)}$ \\
			RELEASE      & $x.(\psi_1 \R \psi_2)$
                   & $\set{x.\psi_1,x.\psi_2} \mid
                      \set{x.\psi_2, x.\X(\psi_1 \R \psi_2)}$ \\
      TRIGGERED    & $x.(\psi_1 \T \psi_2)$
                   & $\set{x.\psi_1,x.\psi_2} \mid
                      \set{x.\psi_2, x.\Y(\psi_1\T\psi_2)}$ \\
		\bottomrule
		\end{tabular}
  \caption{Expansion rules.}
  \label{tab:expansion-rules}
\end{table}

When a poised leaf is obtained, the search can proceed to the next temporal
state. The formulae labelling the current state are used to determine the label of
the next one. Moreover, an amount of time has to be guessed to choose the
timestamp of the next state. This operation is performed by the \Rule{step}
rule.

\begin{rules}
  \item[Step] Let $u$ be a poised node, and let $\delta_{z.\phi}\ge0$ be the
        bound as computed in \cref{prop:delta-bounded-models}. Then,
        $\delta_{z.\phi}+1$ children nodes $u_0,\ldots,u_{\delta_{z.\phi}}$ are
        added to $u$, such that:
        \begin{equation*}
          \begin{aligned}
            \Gamma(u_\delta)&=\set{x.\psi^\delta \mid x.\X\psi\in\Gamma(u)}\\
            \time(u_\delta)&=\time(u)+\delta
          \end{aligned}\quad\text{for all $0\le \delta \le\delta_{z.\phi}$}
        \end{equation*}
\end{rules}

The \Rule{step} rule is one of the most evident differences between the tableau
system for \TPTL and \TPTLbP, and those for \LTL and \LTLP,
since here we have to handle the advancement of the timestamp of the next state.
The formulae in the subsequent state, which are taken from the \emph{tomorrow}
formulae of the current one, are shifted accordingly.

Besides the children added to by the \Rule{step} rule, others can be
subsequently added to a poised node, as it will be shown later, if it does not
fulfil some \emph{past} request coming from the next state. Given a branch
$\bar u=\seq{u_0,\ldots,u_n}$ and a poised node $u_i$, with $0\le i < n$, $u_i$
is said to be a \emph{step node} for the branch $\bar u$ if its child $u_{i+1}$
has been added by the \Rule{step} rule. Moreover, if $u_i$ is a step node for
the branch $\bar u$, we define $\Delta(u_i) = \bigcup_{j < k \le i}
\Gamma(u_k)$, where $u_j$ is the closest step node among the proper ancestors
of $u_i$.

In any case, \emph{before} applying the \Rule{Step} rule to advance to the next
state, the branch has to be checked for contradictions and any other condition that
can cause it to be rejected or accepted. To this end, the following \emph{termination
rules} are applied to poised leaves.
In what follows, any
formula $x.\psi\in\C(\phi)$ of the form $x.\X(\psi_1\U\psi_2)$ is called an
$\X$-eventuality.
Let $\bar u=\seq{u_0,\ldots,u_n}$ be a branch of the tableau.
An $\X$-eventuality $x.\psi = x.\X(\psi_1\U\psi_2)$ is said to be
\emph{requested} at position $i$ if $x.\psi\in\Gamma(u_i)$, and \emph{fulfilled}
at position $j>i$ if $x.\psi_2^{\delta_j}\in\Gamma(u_j)$ and
$x.\psi_1^{\delta_k}\in\Gamma(u_k)$, for all $i < k < j$, where
$\delta_l=\time(u_l)-\time(u_i)$, for $i < l \leq j$.\fitpar

\begin{rules}
  \item[Contradiction] Let $u$ be a poised leaf. If
        $\set{x.p,x.\neg p}\subseteq\Gamma(u)$, for some $p\in\AP$, then $u$ is
        crossed and the branch is rejected.
  \item[Empty] Let $u$ be a poised leaf such that $\Gamma(u)=\emptyset$. Then,
        $u$ is ticked and the branch is accepted.
  \item[Sync] Let $u$ be a poised node. If either $x.(x\le x + c)\in\Gamma(u)$,
        $x.\neg(x\le x+c)\in\Gamma(u)$, $x.(x \equiv_m x + c)\in\Gamma(u)$, or
        $x.\neg(x\equiv_m x + c)\in\Gamma(u)$, but, respectively, $c < 0$,
        $c\ge 0$, $c\not\equiv_m0$, or $c\equiv_m0$, then $u$ is crossed and the
        branch is rejected.
  \item[Yesterday] Let $v$ be a poised leaf such that $x.\Y\psi\in\Gamma(v)$
        for some $x.\psi\in\C(z.\phi)$. If $v$ is the first \emph{step node} of
        its branch, then it is crossed and the branch is rejected.
        Otherwise, let $u<v$ be the closest \emph{step node} among the proper
        ancestors of $v$, $\delta_{u,v}=\time(v)-\time(u)$, and
        $\Omega=\set{x.\psi^{-\delta_{u,v}}\mid x.\Y\psi\in\Gamma(v)}$.
        If $\Omega\not\subseteq\Delta(u)$, then $v$ is crossed, the
        branch is rejected, and a child $u'$ is added to $u$ such that
        $\Gamma(u')=\Gamma(u)\cup\Omega$.
  \item[Loop] Let $v$ be a poised leaf, and $u<v$ a \emph{step node}, proper
        ancestor of $v$, such that $\Gamma(u)=\Gamma(v)$ and all the
        $\X$-eventualities requested in $u$ are fulfilled between $u$ and $v$
        (included). Then,
        \begin{rules}
          \item[Loop_1] if $\time(u)=\time(v)$, then $v$ is crossed and the
                branch rejected;
          \item[Loop_2] if $\time(u)<\time(v)$, then $v$ is ticked and the
                branch accepted.
        \end{rules}
  \item[Prune] Let $w$ be a poised leaf. If there exist three
        \emph{step nodes} $u<v<w$ such that $\Gamma(u)=\Gamma(v)=\Gamma(w)$, and
        each $\X$-eventuality requested in $u$ and fulfilled between $v$ and $w$
        is also fulfilled between $u$ and~$v$, then, $w$ is crossed and the
        branch rejected.
\end{rules}

The above rules resemble the structure of the one-pass and tree-shaped tableau
for \LTLP presented in \cite{GiganteMR17}, but adapted to the new logic. The
\Rule{Sync} rule has been added to the termination rules to detect contradictory
timing constraints. The \Rule{Step} rule, thanks to the temporal shift operator,
can push freeze quantifiers to the next state, without explicitly keeping track
of variable bindings. In such a way, it ensures that  nodes are labelled only by
closed formulae of the form $x.\psi$, the base case of timing constraints
consisting only in formulae of the form $x.(x\sim x + c)$, which involve a
single variable. Judging the validity of the constraints is then trivial. This
mechanism was originally exploited in the graph-shaped
tableau for \TPTL given in \cite{AlurH94}. The \Rule{Loop} rule handles the case
where the branch is cycling through a segment which fulfils all the requests,
and thus a satisfying model of the formulae has been found. However, since timed
state sequences must satisfy the \emph{progress} property, the rule has to
reject those branches where the loop has not advanced  in time (\Rule{Loop_1})
and to accept a branch only if some progress has been made (\Rule{Loop_2}).
In \cref{fig:example1}, we give a brief example of tableau for the TPTL formula
$x.\G y.(p \to y \le x+2)$, which expresses the property that $p$ holds only on
states with timestamp less than $2$. Firstly we focus on node $u_2$: it is crossed
by the \Rule{Loop_1} rule because there is another node (\ie $u_1$) such that all
the conditions of the \Rule{Loop} rule are satisfied but time does \emph{not}
increase between these two nodes. Nevertheless, if we choose to increment by one
time unit the candidate model by means of the \Rule{STEP} rule, we eventually
reach node $u_3$, which does not contain any timed constraint, since they all have
been simplified by the temporal shift $\cdot^\delta$. Now the \Rule{Loop_2} can be
applied on node $u_4$, since nodes $u_3$ and $u_4$ have the same label, all the
X-eventualities (there are none) are fulfilled in between, and the time between
$u_3$ and $u_4$ \emph{does} increase: thus, we tick $u_4$ and accept the
corresponding branch. This, in turn, corresponds to a correct model of the input
formula which starts from the root of the tableau, goes down to $u_4$ and then
cycles between $u_3$ and $u_4$.\fitpar

\begin{figure}[h]
  \includegraphics[width=\textwidth]{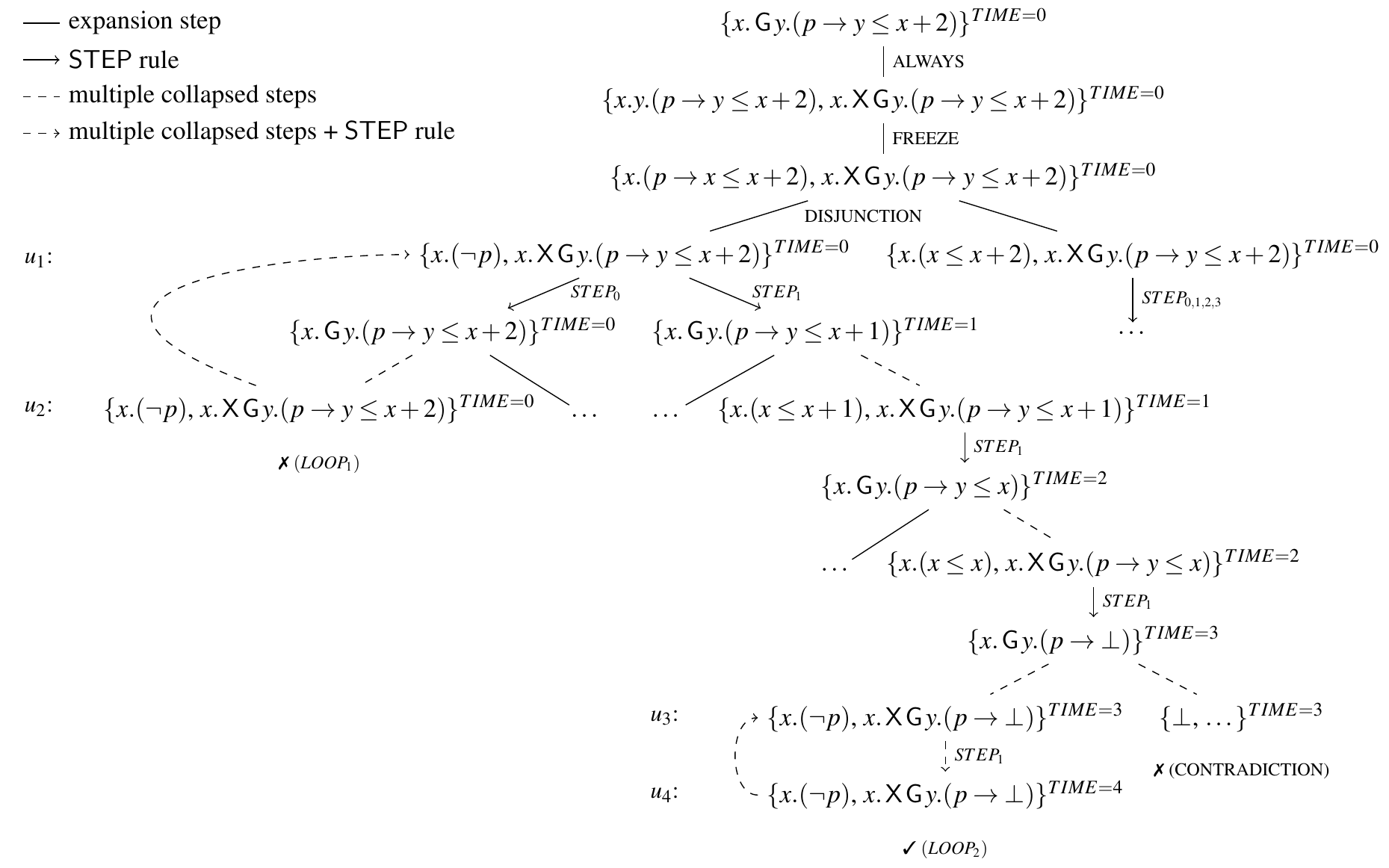}
  \caption{The tableau for the formula $x.\G y.(p \to y \le x+2)$}
  \label{fig:example1}
\end{figure}

The \Rule{Prune} rule handles the case where the branch is
cycling without being able to fulfil all the requests, possibly because some of
them are unsatisfiable. This rule was the main novelty of the one-pass and
tree-shaped tableau system for \LTL in \cite{Reynolds16a},
and, notably, it does not need to be changed at all to work for \TPTLP as well.
An interesting example showing an application of the \Rule{Prune} rule in the
context of a tableau for LTL is shown in \cite{BertelloGMR16}.

Note that, supposing to employ a proper temporal shift operator, the set of
\cref{def:closure} is finite. This fact allows us to prove the termination
of the construction of the tableau with a simple argument.

\begin{theorem}[Termination of tableau construction]
  \label{thm:tptl-termination}
  Let $z.\phi$ be a closed \TPTLP formula and let   $\cdot^\delta$ be a proper temporal shift operator. Then, the construction of a complete tableau
for $z.\phi$, built with $\cdot^\delta$, always terminates in a finite number of steps.
\end{theorem}
\begin{proof}
  First, observe that the tableau 
  for a \TPTLP formula $z.\phi$
  has a finite branching factor, since all the expansion rules create at most
  two children for any node, and the number of children created by the
  \Rule{Step} rule is bounded by $\delta_\phi$. Other children may be added to a
  poised node by failed instances of the \Rule{Yesterday} rule, but since
  $\C(z.\phi)$ is finite, the number of possible different labels is finite, and
  since the rule never creates two nodes with the same label, then the number of
  children added in this way is finite as well. Thus, by König's lemma, for the
  construction to proceed forever the tree should contain at least one infinite
  branch. However, since the number of possible labels is finite, two nodes with
  the same label are guaranteed to appear, and if they do not trigger the
  \Rule{Loop} rule, then, after a finite number of repetitions of the same
  label, the \Rule{Prune} rule is guaranteed to be eventually triggered because
  the different combinations of $\X$-eventualities satisfied between either two
  of those nodes is finite as well.
\end{proof}

\subsection{The tableau system for \TPTL}
\label{sub:tptl-tableau}


Let us now specialise the above general rules to \TPTL formulae. Basically, we need to define a proper temporal shift operator.
Consider a formula $x.\psi\in\C(z.\phi)$, and any other
variable $y$ appearing in $\psi$. Since $x.\psi$ is a closed formula, $y$ must be
quantified inside $\psi$, and, being $\psi$
a future-only formula, it can only be bound to a timestamp greater than or equal to
$x$. Hence, any timing constraint of the form $x\le y + c$, with $c\ge0$, always
holds regardless of the specific evaluation of the variables. A similar
consideration can be made for timing constraints of the form $y\le x + c$, with
$c<0$, which are always false. This fact, originally observed in \cite{AlurH94},
leads to the following definition of the temporal shifting operator for \TPTL
formulae.

\begin{definition}[Temporal shift operator for \TPTL formulae~\cite{AlurH94}]
  \label{def:tptl-shift}
  Let $x.\psi$ be a \emph{closed} \TPTL formula and $\delta\in\N$. 
  Then, $x.\psi^\delta$ is the formula obtained by applying the following steps:
  \begin{enumerate}
    \item replace any timing constraint of the forms $x\le y + c$, $y\le x+c$, and
          $x\equiv_m y+c$, for any other variable $y\in \V$, by, respectively,
          $x\le y + c'$, $y\le x + c''$, and $x\equiv_m y + (c'\mod m)$,
          where $c'=c+\delta$ and $c''=c-\delta$; and then
    \item replace any timing constraint of the forms $x\le y + c'$ and
          $y\le x + c''$, with $c'\ge0$ and $c''<0$, by, respectively, $\true$
          and $\false$.
  \end{enumerate}
\end{definition}

The one-pass and tree-shaped tableau system for \TPTL is obtained from the set of
rules of Section \ref{subsec:skeleton} by considering the temporal shift operator of
\cref{def:tptl-shift}. It can be easily checked that \cref{def:tptl-shift} satisfies
the requirements of \cref{def:temporal-shift} for any non-negative $\delta\ge0$.
Since the \Rule{Yesterday} rule never comes into play with \TPTL
formulae, this is sufficient, as the proofs in \cref{sec:proofs} will confirm.

\subsection{The tableau system for \TPTLbP}
\label{sub:tptlbp-tableau}

Let us now specialise the above set of tableau rules to \TPTLbP.
\TPTLbP is not a proper fragment of \TPTLP \emph{as-is}, and thus it may seem that
those rules cannot be directly applied to \TPTLbP formulae. However, \TPTLbP
can be embedded into a \emph{guarded} fragment of \TPTLP, that is, a syntactic
fragment of the logic, that we call \GTPTLP, where each occurrence of any temporal
operator is guarded by an additional formula which implements the bounded
semantics of \TPTLbP operators.
%
%
\GTPTLP syntax is defined as
follows:
\begin{align*}
  \phi :=
    p \mid {} & \neg\phi_1 \mid \phi_1\lor\phi_2
             \mid x\le y + c \mid x\le c \mid x\equiv_m y + c \\
    \mid {} & x.\X y.(\gamma_w^{x,y} \land \phi_1) \mid
              x.\X y.(\gamma_w^{x,y} \to \phi_1) \mid
              x.\Y y.(\gamma_w^{x,y} \land \phi_1) \mid
              x.\Y y.(\gamma_w^{x,y} \to \phi_1) \\
    \mid {} & x.\bigl(z.(\gamma_w^{x,z}\implies\phi_1) \U
                y.(\gamma_w^{x,y} \land \phi_2)\bigr) \mid
              x.\bigl(z.(\gamma_w^{x,z}\land\phi_1)\R
                y.(\gamma_w^{x,y}\implies\phi_2)\bigr)\\
    \mid {} & x.\bigl(z.(\gamma_w^{x,z}\implies\phi_1) \S
                y.(\gamma_w^{x,y} \land \phi_2)\bigr) \,\mid
              x.\bigl(z.(\gamma_w^{x,z}\land\phi_1) \T
                y.(\gamma_w^{x,y}\implies\phi_2)\bigr),
\end{align*}

\noindent where $\gamma_w^{x,y} = y\le x+w$, if $w\ne+\infty$, and
$\gamma_w = \true$ otherwise, with $w\in\N\cup\set{+\infty}$ and $x$ and $y$
fresh in $\phi_1$ and $\phi_2$. Moreover, as in \TPTLbP, each temporal operator
can appear with $w=+\infty$ only if the corresponding formula is closed. All the
temporal operators where $w\ne+\infty$ are called \emph{guarded}.

One can check that
\begin{claims}
  \item the \emph{negated normal form} of a \GTPTLP formula is still a \GTPTLP
        formula, and
  \item each \TPTLbP formula can be translated into an equivalent \GTPTLP one.
\end{claims}
A notable example is the translation of the $\X$ and $\wX$ operators (and,
symmetrically, $\Y$ and $\wY$), that both get translated into a formula using a
guarded $\X$ operator, but with the guard that, respectively, is conjuncted to
and implies the target formula, \ie $\X_w\psi\equiv x.\X y.(y \le x + w \land
\psi)$ and $\wX_w\psi\equiv x.\X y.(y\le x + w \implies \psi)$, if $w\ne+\infty$,
and simply $\X_{+\infty}\psi\equiv\wX_{+\infty}\psi\equiv\X\psi$ otherwise. The
translation provides a sound and complete embedding of \TPTLbP into (the \GTPTLP
syntactic fragment of) \TPTLP.

\begin{lemma}
  Let $\phi$ be a \TPTLbP formula over the proposition letters $\AP$ and the
  variables $\V$. Then, there exists a \GTPTLP formula $\phi'$ such that for
  any timed state sequence $\rho$, any environment $\xi$, and any $i\ge0$, it
  holds that $\rho\models^i_\xi\phi$ if and only if $\rho\models^i_\xi\phi'$.
\end{lemma}

Hence, we can apply the general tableau rules to the \GTPTLP translation of any
\TPTLbP formula, provided that, similar to the \TPTL case, a proper temporal
shift operator can be defined. This can actually be done by exploiting the following
observation:
thanks to the bounds applied to the \TPTLbP temporal operators, whose semantics
is implemented in \GTPTLP formulae by means of the guards, when interpreting a
timing constraint like $x\le y + c$, the distance between variables $x$ and $y$
cannot be greater than an upper bound $W$ that depends on the bounds applied to
the temporal operators of the formula. This observation was exploited in
\cite{DellaMonicaGMSS17} to prove decidability and \EXPSPACE-completeness of
\TPTLbP. Now, given a \GTPTLP formula $z.\phi$, let $m$ be the number of
\emph{guarded} temporal operators used in $z.\phi$, let $w_0=\max\set{w_1,
\ldots,w_m, \delta_{z.\phi}}$, where $w_1,\ldots,w_m$ are the bounds applied
to the respective guarded temporal operators and $\delta_{z.\phi}$ is computed
as per \cref{prop:delta-bounded-models}, and let $W_{z.\phi}=w_0\cdot(m + 1)$.

\begin{definition}[Temporal shift operator for \GTPTLP~\cite{DellaMonicaGMSS17}]
  \label{def:tptlbp-shift}
  Let $z.\phi$ be a \emph{closed} \TPTL formula, $\delta\in\N$, and $x.\psi\in\C(z.\phi)$.
  Then, $x.\psi^\delta$ is the formula obtained by applying the following steps:
  \begin{enumerate}
    \item replace any timing constraint of the forms $x\le y + c$, $y\le x+c$, and
          $x\equiv_m y+c$, for any other variable $y\in V$, by, respectively,
          $x\le y + c'$, $y\le x + c''$, and $x\equiv_m y + (c'\mod m)$,
          where $c'=c+\delta$ and $c''=c-\delta$; and then
    \item replace any timing constraint of the forms $x \le y+c$ and $y \le x+c$
          either by $\true$, if $c \ge W_{z.\phi}$, or by $\false$, if
          $c < -W_{z.\phi}$.
  \end{enumerate}
\end{definition}

It can be easily shown that \cref{def:tptlbp-shift} defines a
temporal shift operator as per \cref{def:temporal-shift}~\cite{DellaMonicaGMSS17} .


\section{Soundness and Completeness}
\label{sec:proofs}

We now prove soundness and completeness of the
tableau systems for \TPTL and \TPTLbP.
Given that the two systems are nearly identical, excepting for the definition of the
proper temporal shift operator, both proofs will be given at once,
differentiating between the two logics only when necessary.

\subsection{Soundness}
Here we prove that the tableau system is \emph{sound}, that is, if a complete
tableau for a formula has a successful branch, then the formula is satisfiable (and a model for the formula can be effectively extracted from the
successful branch). As a preliminary step, we introduce the notion of \emph{pre-model}: an abstract,
easy to manipulate representation of a model of a formula.

\begin{definition}[Atom]
	\label{def:atom}
	An \emph{atom} for a \TPTL~/~\TPTLbP formula $z.\phi$ is a set
	$\Delta\subseteq\C(z.\phi)$ such that:

	\begin{conditions}
		\item $x.p \in \Delta$
				& $x.\neg p\not\in\Delta$, for any proposition $x.p\in\C(z.\phi)$;
		\item $x.y.\psi_1 \in \Delta$
				& $x.\psi_1[y/x] \in \Delta$;
		\item $x.(\psi_1 \land \psi_2) \in \Delta$
				& $\set{x.\psi_1, x.\psi_2} \subseteq \Delta$;
		\item $x.(\psi_1 \lor \psi_2) \in \Delta$
				& either $x.\psi_1 \in \Delta$ or $x.\psi_2 \in \Delta$;
		\item $x.(\psi_1 \U \psi_2) \in \Delta$
				& either $x.\psi_2 \in \Delta$ or
									 $\set{x.\psi_1, x.\X(\psi_1 \U \psi_2)} \subseteq \Delta$;
		\item $x.(\psi_1 \R \psi_2) \in \Delta $
				& either $\set{x.\psi_1,x.\psi_2} \subseteq \Delta$ or
								 $\set{x.\psi_2,x.\X(\psi_1 \R \psi_2)} \subseteq \Delta$;
		\item $x.(\psi_1 \S \psi_2) \in \Delta$
				& either $x.\psi_2 \in \Delta$ or
								 $\set{x.\psi_1, x.\Y(\psi_1 \S \psi_2)} \subseteq \Delta$;
		\item $x.(\psi_1 \T \psi_2) \in \Delta$
				& either $\set{x.\psi_1,x.\psi_2} \subseteq \Delta$ or
								 $\set{x.\psi_2, x.\Y(\psi_1 \T \psi_2)} \subseteq \Delta$.
	\end{conditions}
\end{definition}

Intuitively, \emph{atoms} are sets of formulae such that the presence of each
non-elementary formula is justified (\ie implied) by the elementary formulae in
the set, and each non-elementary formula that can be justified by the set is
present.

\begin{definition}[Pre-model]
	\label{def:pre-model}
	Let $z.\phi$ be a closed \TPTL \ / \TPTLbP formula. A \emph{pre-model} of
	$z.\phi$ is a pair $\Pi = \pair{\bar\Delta,\bar\iota}$, where
	$\bar\iota=\seq{\iota_0,\iota_1,\dots}$ is an infinite sequence of timestamps
	satisfying the \emph{progress} and \emph{monotonicity} conditions, and
	$\bar\Delta=\seq{\Delta_0,\Delta_1,\ldots}$ is an infinite sequence of atoms
	for $z.\phi$ such that, for all $i\ge0$,\fitpar
	\begin{enumerate}
		\item \label{def:pre-model:base}
					$z.\phi \in \Delta_0$;
		\item \label{def:pre-model:tomorrow}
					if $x.\X\psi \in \Delta_i$, then
					$x.\psi^{\delta_{i+1}}\in\Delta_{i+1}$;
		\item \label{def:pre-model:until}
					if $x.(\psi_1 \U \psi_2) \in \Delta_i$, then there exists a $j\ge i$
					such that $x.\psi_2^{\delta_{i,j}} \in \Delta_j$ and
					$x.\psi_1^{\delta_{i,k}}\in\Delta_k$ for all $i \le k < j$;
		\item \label{def:pre-model:yesterday}
					if $x.\Y\psi\in\Delta_i$, then $i>0$ and
					$x.\psi^{-\delta_i}\in\Delta_{i-1}$;
		\item \label{def:pre-model:since}
					if $x.(\psi_1 \S \psi_2) \in \Delta_i$, then there exists a $j \le i$
					such that $x.\psi_2^{-\delta_{j,i}} \in \Delta_j$ and
					$x.\psi_1^{-\delta_{k,i}}\in\Delta_k$ for all $j < k \le i$,
	\end{enumerate}
	where $\delta_0=\iota_0$, $\delta_{i+1} = \iota_{i+1} - \iota_i$ for $i\ge0$,
	and $\delta_{n,m}=\sum_{n< p\le m}\delta_p$ for all $n,m\in\N$.
\end{definition}

Pre-models take their name from the fact that they abstractly represent a model
for their formula, and thus the existence of a pre-model witnesses the
satisfiability of the formula.

\begin{lemma}
	\label{lemma:pre-model-to-model}
	Let $z.\phi$ be a closed \TPTL \ / \TPTLbP formula. If $z.\phi$ has a
	pre-model, then $z.\phi$ is satisfiable.
\end{lemma}
\begin{proof}
	Let $\Pi=(\bar\Delta,\bar\iota)$ be a pre-model of $z.\phi$ and let
	$\rho=(\sigma,\tau)$ be a timed state sequence such that $\iota_i=\tau_i$ and
	$x.p\in\Delta_i$ if and only if $\rho\models^i p$. Note that each $\tau$
	satisfies the monotonicity and progress conditions because $\bar\iota$ does by
	definition of pre-model. Then, we show that $\rho\models z.\phi$ and thus the
	formula is satisfiable.

	For any $x.\psi\in\C(z.\phi)$, let the \emph{nesting degree} $\deg(x.\psi)$ of
	$x.\psi$ be defined inductively as follows: $\deg(x.p)=\deg(x.\neg p)=0$ for
	$p\in\AP$, $\deg(x.y.\psi)=\deg(y.\psi)+1$, and $\deg(x
    (\phi_1\circ\phi_2))=\max(\deg(x.\psi_1),$ $\deg(x.\psi_2))+1$,
    with $\circ \in \{\land,\lor,\U,\S,\R,\T\}$.
	We prove by induction on $\deg(x.\psi)$ that if $x.\psi\in\Delta_i$, then
	$\rho\models^i\psi$ for any $x.\psi\in\C(z.\phi)$ and any $i\ge 0$ (since
	all $x.\psi\in\C(z.\phi)$ are closed, we do not need to take care of
	environments). The thesis then follows from
	\cref{def:pre-model:base} of \cref{def:pre-model}, since $z.\phi \in
	\Delta_0$.

    As for the base case, if $x.p\in\Delta_i$ or $x.\neg p\in\Delta_i$, then the
	thesis follows by the definition of $\rho$.

    As for the inductive step, we go by cases:

	\begin{enumerate}
		\item \label{prop:models:base}
					if $x.y.\psi\in\Delta_i$, then $x.\psi[y/x]\in\Delta_i$ and by the
					inductive hypothesis $\rho\models^i x.\psi[y/x]$, thus
					$\rho\models^i x.y.\psi$;
		\item \label{prop:models:connectives}
					if $x.(\psi_1\lor\psi_2)\in\Delta_i$ (resp., $x.(\psi_1\land\psi_2)$),
					then by definition of atom and the inductive hypothesis, either
					$\rho\models^i x.\psi_1$ or $\rho\models^i x.\psi_2$ (resp., both), and thus $\rho\models^i x.(\psi_1\lor\psi_2)$
					(resp.,~$\rho\models^i x.(\psi_1\land\psi_2)$);
		\item \label{prop:models:tomorrow}
					if $x.\X\psi\in\Delta_i$, then, by \cref{def:pre-model:tomorrow} of
					\cref{def:pre-model}, it holds that
					$x.\psi^{\delta_{i+1}}\in\Delta_{i+1}$. Since
					$\deg(x.\psi)<\deg(x.\X\psi)$, by the inductive hypothesis it
					follows that $\rho\models^{i+1}_\xi x.\psi^{\delta_{i+1}}$, for any
					$\xi$. By \cref{def:temporal-shift}, this implies that
					$\rho\models^{i+1}_{\xi[x\takes\tau_{i+1}-\delta_{i+1}]}\psi$, that
					is, $\rho\models^{i+1}_{\xi[x\takes\tau_i]}\psi$. Then, by the
					semantics of the \emph{tomorrow} operator and of the freeze
					quantifier, we have $\rho\models^i_{\xi[x\takes\tau_i]}\X\psi$ and
					thus $\rho\models^i_\xi x.\X\psi$;
		\item \label{prop:models:until}
					if $x.(\psi_1\U\psi_2)\in\Delta_i$, then, by definition of atom,
					there exists $j\ge i$ such that $x.\psi_2^{\delta_{i,j}} \in
					\Delta_j$ and $x.\psi_1^{\delta_{i,k}}\in\Delta_k$, for all $i\le k<j$.
					Then, by the inductive hypothesis, $\rho\models^j_\xi
					x.\psi_2^{\delta_{i,j}}$ and $\rho\models^k_\xi
					x.\psi_1^{\delta_{i,k}}$, for any $\xi$ and all $i\le k<j$. By
					\cref{def:temporal-shift}, we have that
					$\rho\models^j_{\xi[x\takes\tau_j-\delta_{i,j}]}\psi_2$ and
					$\rho\models^k_{\xi[x\takes\tau_k-\delta_{i,k}]}\psi_1$ for all
					$i\le k<j$, that is, $\rho\models^j_{\xi[x\takes\tau_i]}\psi_2$ and
					$\rho\models^k_{\xi[x\takes\tau_i]}\psi_1$ for all $i\le k<j$.
					Finally, by the semantics of the \emph{until} operator and of the
					freeze quantifier, we have
					$\rho\models^i_{\xi[x\takes\tau_i]}\psi_1\U\psi_2$ and thus
					$\rho\models^i_\xi x.(\psi_1\U\psi_2)$;\fitpar
		\item \label{prop:models:others}
					the case when $x.(\psi_1\R\psi_2)\in\Delta_i$ is similar to
					\cref{prop:models:until}, and the cases when $x.\Y\psi\in\Delta_i$,
					$x.(\psi_1\S\psi_2)\in\Delta_i$ or $x.(\psi_1\T\psi_2)\in\Delta_i$ are
					similar and specular to \cref{prop:models:tomorrow,prop:models:until},
					respectively.\qedhere
	\end{enumerate}
\end{proof}

To complete the proof, it suffices to show that a pre-model for a
formula can be obtained from a successful branch of the tableau.

\begin{lemma}
	\label{lemma:branch-to-pre-model}
	Let $z.\phi$ be a closed \TPTL or \TPTLbP formula and $T$ a complete tableau
	for $z.\phi$. If $T$ has a successful branch, then there exists a
	pre-model for $z.\phi$.
\end{lemma}
\begin{proof}
	Let $\bar u=\seq{u_0,\ldots,u_n}$ be a successful branch of $T$ and let
	$\bar\pi=\seq{\pi_0,\ldots,\pi_m}$ be the subsequence of step nodes of $\bar
	u$. Intuitively, a pre-model for $z.\phi$ can be obtained from $\bar u$ by
	building the atoms from the labels of the step nodes, and extending them to an
	infinite sequence. Let $\Delta(\pi_i)$ be the atom obtained from $\Gamma(\pi_i)$
    by arbitrarily completing it with missing literals and closing it over the
	requirements of \cref{def:atom}. The sequence of
	$\Delta(\pi_i)$, with $0\le i\le m$, forms the basic skeleton of the pre-model
	$\Pi=(\bar\Delta,\bar\iota)$ defined as follows. As for the atoms,
	$\Delta_i=\Delta(\pi_{\K(i)})$, where $\K: \N \to \set{0,\ldots,m}$, is defined
	differently depending on which rule caused the branch to be accepted:
	\begin{enumerate}
	\item if $\pi_m$ was ticked by the \Rule{Loop_2} rule, then there
				exists $k<m$ such that $\Gamma(\pi_k)=\Gamma(\pi_m)$ and all
				the $\X$-eventualities requested in $\pi_k$ are fulfilled
				between $\pi_k$ and $\pi_m$. Then, the pre-model repeats
				forever the atoms between $\Delta(\pi_{k+1})$ and $\Delta(\pi_m)$, and
				thus $\K(i)= i$, for $0 \leq i < k$, and $\K(i)=k+((i-k) \mod T)$, with
                $T=m-k$, for $i \geq k$;
	\item if $\pi_m$ was ticked by the \Rule{Empty} rule, then
				$\Gamma(\pi_m)=\emptyset$ and the pre-model repeats forever the
				atom $\Delta(\pi_m)$, hence $\K(i)=i$ if $i< m$, and $\K(i)=m$ if
				$i\ge m$.
	\end{enumerate}
	As for the sequence of timestamps, it is taken directly from the step nodes
	accordingly:
	\begin{enumerate}
	\item if $\pi_m$ was ticked by the \Rule{Loop_2} rule, then
				$\iota_i=\time_{i-1}+(\time(\pi_{\K(i)})-\time(\pi_{\K(i)-1}))$ for all
				$i\ge0$;
	\item if $\pi_m$ was ticked by the \Rule{Empty} rule, then
				$\iota_i=\time(\pi_i)$ for $i\le m$, and $\iota_{i+1}=\iota_i+1$
				for all $i > m$.
	\end{enumerate}

	We now show that $\Pi$ is indeed a pre-model for $z.\phi$.
	First, note that, by construction, $\bar\iota$ satisfies the \emph{progress} and
	\emph{monotonicity} conditions (in particular, \Rule{Loop_2} rule ensures that
    $\time(\pi_m)>\time(\pi_k)$. Then, observe that $z.\phi\in\Delta_0$ because
    $z.\phi\in\Gamma(\pi_0)$ by construction, and thus \cref{def:pre-model:base} of
    \cref{def:pre-model} is satisfied.

	Consider now any formula $x.\X\psi\in\Delta_i$. Being an
	elementary formula, we know that $x.\X\psi\in\Gamma(\pi_{\K(i)})$. Two cases
	have to be considered. If $\pi_{\K(i+1)}=\pi_{\K(i)+1}$, \ie the next atom
	comes from the actual successor of the current one in the tableau branch, then,
	by the \Rule{Step} rule, $x.\psi^{\delta_{i+1}}\in\Delta_{i+1}$.
	Otherwise, $\Delta_i=\Delta(\pi_m)$ and $\pi_m$ was ticked by the
	\Rule{Loop_2} (because $\Delta_i$ is not empty), and thus
	$\Delta_{i+1}=\Delta(\pi_{k+1})$ for some $k<m$ such that
	$\Gamma(\pi_k)=\Gamma(\pi_m)$. Hence, $x.\X\psi \in \Gamma(\pi_k)$
	as well, and, by the \Rule{Step} rule applied to $\pi_k$,
	$x.\psi^{\delta_{i+1}} \in \Delta(\pi_{k+1})=\Delta_{i+1}$, and thus
	\cref{def:pre-model:tomorrow} of \cref{def:pre-model} is satisfied.

	Finally, consider any formula $x.\Y\psi\in\Delta_i$ and thus
	$x.\Y\psi\in\Gamma(\pi_{\K(i)})$. By the \Rule{Yesterday} rule,
	$i>0$. As in the previous case, either $\Delta_{i-1}$ is the atom
	coming from the previous step node, and thus
	$x.\psi^{-\delta_i}\in\Delta_{i-1}$ by the \Rule{Yesterday} rule, or
	$\Delta_i=\Delta(\pi_{k+1})$ for some $k$ that triggered the \Rule{Loop_2}
	rule because $\Gamma(\pi_k)=\Gamma(\pi_m)$. By the \Rule{Yesterday} rule,
	$x.\Y\psi^{-\delta_{k+1}}\in\Delta_k$, and, since
	$\delta_{k+1}=\delta_i$, $x.\Y\psi^{-\delta_i}\in\Delta_m=\Delta_{i-1}$.

    The other cases
	are straightforward in view of how expansion rules are defined.
\end{proof}

\begin{theorem}[Soundness]
	Let $z.\phi$ be a closed \TPTL \ / \TPTLbP formula, and let $T$ be a complete
	tableau for $z.\phi$. If $T$ has a successful branch, then $z.\phi$ is
	satisfiable.
\end{theorem}
\begin{proof}
	Extract a pre-model for $z.\phi$ from the successful branch of $T$ as in
	\cref{lemma:branch-to-pre-model}, and then obtain from it an actual model for
	the formula as in \cref{lemma:pre-model-to-model}.
\end{proof}

\subsection{Completeness}
We now prove the completeness of the tableau system, i.e., if a formula
$z.\phi$ is satisfiable, then any complete tableau $T$ for it
has an accepting branch. We make use of a new model-theoretic argument providing a much
simpler and shorter proof, which sidesteps the complex combinatorial argument used
in completeness proofs for the one-pass tree-shaped tableaux for
\LTL~\cite{Reynolds16a} and \LTLP~\cite{GiganteMR17}.

To start with, we introduce the key concept of \emph{greedy pre-model}. Given a pre-model
$\Pi=\pair{\bar\Delta,\bar\iota}$, an $\X$-eventuality
$x\psi = x.\X(\psi_1\U\psi_2)$ is \emph{requested} at position $i\ge0$ if
$x.\psi\in\Delta_i$, and \emph{fulfilled} at $j > i$ if $j$ is the
first position where $x.\psi_2^{\delta_{i,j}}\in\Delta_j$ and
$x.\psi_1^{\delta_{i,k}}\in\Delta_k$, for all $i<k<j$. Let
$\E(z.\phi)=\{x.\psi\in\C(z.\phi)\mid\text{$x.\psi$ is an $\X$-eventuality}\}$.
For each position $i\ge0$, we define the \emph{delay vector at position $i$} as a
function $\delay_i:\E(z.\phi)\to\N$ providing a natural number for each
eventuality in $\E(z.\phi)$, as follows:
\begin{equation*}
 \delay_i(x.\psi) =
 \begin{cases}
   0 & \text{if $x.\psi$ is not requested at position $i$}\\
   n & \text{if $x.\psi$ is requested at $i$ and fulfilled at $j$
	 such that $n=j-i$}
 \end{cases}
\end{equation*}

Intuitively, $\delay_i(x.\psi)$ is the number of states
elapsed between the request and the fulfilment of $x.\psi$. We
denote as $\bar\delay=\seq{\delay_0,\delay_1,\ldots}$ the sequence of delay
vectors of the atoms of $\bar\Delta$, and define $\delay_i\preceq \delay_i'$ if and only if
$\delay_i(\psi)\le \delay_i'(\psi)$, for all $\psi\in\E(\phi)$. A pre-order
relation on pre-models of a given formula can be defined by comparing the
$\delay_i$ lexicographically: $\Pi\preceq\Pi'$ if $\delay_0 < \delay_0'$ or
$\delay_0=\delay_0'$ and $\Pi_{\ge1}\preceq\Pi'_{\ge1}$, where
$\Pi_{\ge1}=\pair{\bar\Delta_{\ge1},\bar\iota_{\ge1}}$ with
$\bar\Delta_{\ge1}=\seq{\Delta_1,\Delta_2,\ldots}$ and
$\bar\iota_{\ge1}=\seq{\iota_1,\iota_2,\ldots}$. Greedy pre-models are minimal
elements of this pre-order. We show that if a formula admits a pre-model, then
it admits a greedy pre-model. The completeness result can then be proved
directly.

\begin{definition}[Greedy pre-models]
	\label{def:greedy-pre-models}
  Let $\Pi$ be a pre-model for a formula $z.\phi$. $\Pi$ is \emph{greedy} if
  there is no pre-model $\Pi' \neq \Pi$ such that $\Pi'\preceq\Pi$.
\end{definition}

\begin{lemma}
	\label{lemma:greedy-pre-models}
	Let $\Pi$ be a pre-model for a formula $z.\phi$. Then, there is a
	\emph{greedy} pre-model $\Pi'\preceq\Pi$.
\end{lemma}
\begin{proof}
	We distinguish two cases. If there is a finite sequence $\Pi_1 \
    (= \Pi) \succ \Pi_2 \succ \ldots \succ \Pi_n$, with $n \geq 1$, which
    is maximal with respect to $\succ$, i.e., it cannot be further extended,
    then $\Pi' = \Pi_n$ is a greedy model with $\Pi'\preceq\Pi$.
    Otherwise, let $\Pi_1 \ (= \Pi) \succ \Pi_2 \succ \ldots$ be
    an infinite sequence of pre-models. We prove that
	its limit
	is
a greedy model $\Pi'$.
    To this end, it suffices to show that for every $n \in \N$ (prefix length), there
    is $m \in \N$ (pre-model index) such that the prefix up to position
    $n$ of pre-models $\Pi_m, \Pi_{m+1}, \ldots$ is the same.

    For $i \geq 1$, let $\delay^i=\seq{\delay^i_0,\delay^i_1,\ldots}$
    be the sequence of delay vectors of $\Pi_i$. Let us consider the $j$-th
    pre-model $\Pi_j$, for some $j \geq 1$. By definition of $\succ$, there
    is a position $n_j \ge 0$ such that $\delay^{j+1}_{n_j} < \delay^j_{n_j}$,
    and $\delay^{j+1}_m = \delay^j_m$, for all $0\le m < n_j$. We show that
    there are finitely many indexes $l > j$ (let $\overline{l}$ be the largest
    one) for which there exists a position $n_k$, with $n_k \leq n_j$, such
    that $\delay^{l+1}_{n_k} < \delay^l_{n_k}$, and $\delay^{l+1}_m =
    \delay^j_m$, for all $0\le m < n_k$.
    We prove it by contradiction. Assume that there are infinitely many. Let $n_h$
    be the leftmost position that comes into play infinitely many times.
    If $n_h = 0$, then there is an infinite strictly decreasing sequence of delay
    vectors $\delay^{h_1}_{0}>\delay^{h_2}_{0}>\delay^{h_3}_{0}>\ldots$, with
    $j < h_1 < h_2 < h_3 < \ldots$, which cannot be the case since the ordered set
    $(\N^{|\E(z.\phi)|},\le)$ is well-founded (the definition of temporal shift
    operators ensures that the closure set of $z.\phi$ is finite, and thus $\E(z.\phi)$
    is finite as well). Let $0 < n_h \leq n_j$. Since the positions
    to the left of $n_h$ are chosen only finitely many times, there exists a tuple
    $(\delay_0, \ldots, \delay_{n_h-1})$ which is paired with an infinite strictly decreasing
    sequence of delay vectors $\delay^{h_1}_{n_h}>\delay^{h_2}_{n_h}>\delay^{h_3}_{n_h}>
    \ldots$, with $j < h_1 < h_2 < h_3 < \ldots$, which again cannot be the case since
    the ordered set $(\N^{|\E(z.\phi)|},\le)$ is well-founded. This allows us to conclude
    that the prefix up to position $n_j$ of all pre-models of index greater than or equal
    to $\overline{l}$ is the same.
\end{proof}

\begin{theorem}[Completeness]
	Let $z.\phi$ be a closed \TPTL \ / \TPTLbP formula and let $T$ be a complete
	tableau for $z.\phi$. If $z.\phi$ is satisfiable, then $T$ contains a
	successful branch.
\end{theorem}
\begin{proof}
	Let $\rho=\pair{\sigma,\tau}$ be a model for $z.\phi$. It is straightforward
    to build a pre-model for $z.\phi$ from $\rho$. Then, given a pre-model for $z.\phi$,
    \cref{lemma:greedy-pre-models} ensures that a \emph{greedy} pre-model for
    it exists. We can thus restrict our attention to greedy pre-models.
    Let $\Pi=\seq{\bar\Delta,\bar\iota}$ be a greedy pre-model for $z.\phi$. We look for a successful branch
    in $T$ by using $\Pi$ as a guide to descend down the tree until a leaf is
    found, showing that any leaf found in this way must be ticked. The descent
    proceeds as follows.
    At each step $i\ge0$, we maintain a sequence of nodes (which will be the prefix
	of some branch of the tree) $\bar u_i=\seq{u_0,u_1,\ldots,u_i}$ that is
	extended to $\bar u_{i+1}=\seq{u_0,u_1,\ldots,u_i,u_{i+1}}$ by choosing
	$u_{i+1}$ among the children of $u_i$. A map $J:\N\to\N$ is built during the
	descent, where initially $J(0)=0$, which links each nodes in $\bar u_i$ to a
	position in the pre-model by maintaining the invariant that if
	$x.\psi\in\Gamma(u_k)$, then $x.\psi\in\Delta_{J(k)}$, for each $0\le k\le i$ and
	each $x.\psi\in\C(z.\phi)$. At each step $i\ge0$, $u_{i+1}$ is chosen among
	the children of $u_i$ in the following way: if $u_i$ is not a poised node,
	$u_{i+1}$ is chosen as any of its children $u_i'$ such that $\Gamma(u_i')$
	satisfies the invariant. It is easy to check that at least one such child
	exists by construction because of how expansion rules are defined and the fact
	that $\Pi$ is a pre-model. If, otherwise, $u_i$ is a poised node, then it has
	$\delta_{z.\phi}$ children $\seq{u^S_0,\ldots, u^S_{\delta_{z.\phi}}}$ created
	by the \Rule{Step} rule, and potentially other children $\seq{u^Y_0,\ldots,
	u^Y_n}$ added by failed instances of the \Rule{Yesterday} rule. If there is
	any $u^Y_{i}$ whose label satisfies the invariant, then one of those is
	selected as $u_{i+1}$. If no such child exists, $u_{i+1}$ is chosen according
	to the timestamp of the next atom in the pre-model, \ie
	$u_{i+1}=u^S_{\delta_{J(i+1)}}$. The invariant in this case is satisfied by
	construction because of the definition of the \Rule{Step} rule.\fitpar

	Since each step always descends down the tree, we will eventually reach a leaf
	$u_n$. We now show that $u_n$ has to be a ticked leaf. If instead
	$u_n$ was crossed, it could not have been crossed by contradiction, because there would
	be some $p$ and $\neg p$ in $\Delta(u_n)$ that would imply
	that $p\in\Delta_{J(n)}$ and $\neg p\in\Delta_{J(n)}$, which cannot be the
	case. Similarly, it could not have been crossed by the \Rule{Synch} rule. Furthermore, the \Rule{Loop_1} rule could not have crossed $u_n$, because the
	timestamps were chosen following $\bar\iota$, which by definition satisfies the
	\emph{progress} and \emph{monotonicity} conditions. Then, $u_n$ has to have been
	crossed by the \Rule{Prune} rule, hence there exist other two nodes $u_m$
	and $u_r$ such that $\Gamma(u_m)=\Gamma(u_r)=\Gamma(u_n)$ and all the
	eventualities requested in $u_m$ and fulfilled between $u_r$ and $u_n$ are
	also fulfilled between $u_m$ and $u_r$, and
	$\Delta_{J(m)}=\Delta_{J(r)}=\Delta_{J(n)}$. Now, it
	can be checked that the pre-model $\Pi'$
	obtained by removing all the atoms between $\Delta_{J(r)+1}$ and $\Delta_{J(n)}$ is still a pre-model for
	$z.\phi$. Then, we show that $\Pi'\preceq\Pi$, leading to a contradiction,
	since we supposed that $\Pi$ was greedy.\fitpar

	We proceed by showing that $\delay'_i\prec \delay_i$, while $\delay'_n\preceq
	\delay_n$ for all $n < i$, thus implying that $\Pi' \prec \Pi$. To
	this end, we need to show that there is at least one $\X$-eventuality
	$x.\psi$ for which
	$\delay'_i(x.\psi)<\delay_i(x.\psi)$ while the other
	values of the delay vector for the other eventualities remains constant.
  First, consider an eventuality $x.\psi$ which is requested in
  $\Delta_i$, but not in $\Delta_j$. Then, it holds that its first fulfilment happens
  before $\Delta_j$ and the cut between $\Delta_j$ and $\Delta_k$ cannot change
  its delay. Now, suppose $x.\psi$ is requested in $\Delta_i$
  and $\Delta_j$ and is fulfilled between $\Delta_j$ and $\Delta_k$. Hence, by
  definition of PRUNE rule, it is also fulfilled between $\Delta_i$ and
  $\Delta_j$, thus again its first fulfilment after $\Delta_i$ is before
  $\Delta_j$, and the cut does not change its delay. The remaining case is that
  of $x.\psi$ being requested in $\Delta_i$ and $\Delta_j$ but not
  fulfilled between them, and thus neither between $\Delta_j$ and $\Delta_k$. At
  least one eventuality of this kind is required to exist by the definition of
  \Rule{Prune} rule. Then, since $x.\psi$ is not fulfilled before
  $\Delta_k$, it must be requested there, and fulfilled later, and the cut
  between $\Delta_j$ and $\Delta_k$ will decrease the value of
  $\delay'_i(x.\psi)$. Thus $\delay'_i\prec \delay_i$. Now,
  consider any position $n < i$. In any of those positions, for any eventuality
  $x.\psi$, $\delay_n(x.\psi)$ cannot increase
  because of the cut, otherwise the first fulfilment of $x.\psi$
  would have been between $\Delta_j$ and $\Delta_k$, which cannot be the case
  because all the eventualities fulfilled there
	are fulfilled also before,
  between $\Delta_i$ and $\Delta_k$. Hence $d'_n=d_n$ for all $n<i$, and thus
	$\Pi'\prec\Pi$.
\end{proof}


\section{Conclusions}
\label{sec:conclusions}

In this paper, we developed one-pass and tree-shaped tableau systems for \TPTL and
\TPTLbP. They extend those for \LTL~\cite{Reynolds16a} and \LTLP~\cite{GiganteMR17}
with the ability of dealing with freeze quantifiers and timing constraints. Notably,
the \Rule{Prune} rule, which was the main novelty of the one-pass and tree-shaped
tableau system for \LTL, did not need to be changed at all to
work in the new systems. This confirms the great
extensibility of this tableau system. The completeness of the \Rule{Prune} rule
has been proved here with a new model-theoretic argument, much simpler than those used in the
proofs for \LTL and \LTLP. Whether or not such a tableau system can be
extended to support full \TPTLP is still an open problem.

\bibliographystyle{eptcs}
\bibliography{biblio}
\end{document}